\documentclass[submission,copyright,creativecommons]{eptcs}
\providecommand{\event}{NCL22} 
\usepackage{breakurl}             
\usepackage{underscore}           

\usepackage{color}
\usepackage{amsthm,amssymb}

\theoremstyle{definition}
 
\newtheorem{definition}{Definition}
\newtheorem{proposition}[definition]{Proposition}

\newtheorem{remark}[definition]{Remark}

\newcommand{\To}{{\to}}

\title{Superconnexivity Reconsidered\thanks{The research by Andreas Kapsner has been supported by the Deutsche Forschungsgemeinschaft (DFG, German Research Foundation), Project 436508789. Moreover, the research by Hitoshi Omori has been supported by a Sofja Kovalevskaja Award of the Alexander von Humboldt-Foundation, funded by the German Ministry for Education and Research. We would like to thank the audience in Bochum for their comments on an earlier version of this paper. In particular, our special thanks goes to Francesco Paoli for extremely fruitful discussions. Finally, but not the least, we would also like to thank the anonymous referees for their supportive, careful and helpful comments.}}
\author{Andreas Kapsner
	\institute{MCMP\\ LMU Munich \\ Munich, Germany}
	\email{Andreas.Kapsner@lrz.uni-muenchen.de}
	\and
	Hitoshi Omori
	\institute{Department of Philosophy I\\ Ruhr Univeristy Bochum\\
		Bochum, Germany}
	\email{Hitoshi.Omori@rub.de}
}
\def\titlerunning{Superconnexivity Reconsidered}
\def\authorrunning{A. Kapsner \& H. Omori}

\begin{document}
\maketitle
	
\begin{abstract}
We reconsider the idea of \textit{superconnexivity}, an idea that has not received much attention so far. We inspect more closely the problems with the proposal that are responsible for this disregard. However, we also suggest a slight modification of the idea that has a much better chance of delivering the desired results, which we call \textit{super-bot-connexivity}. 
\end{abstract}
	
\section{Introduction}
	
In his paper “Strong Connexivity” (cf. \cite{Kapsner2012strong}), Andreas Kapsner claimed that the connexive principles were based on intuitions that made an amendment necessary. Connexive logic, up to that point, had largely been considered to be exhaustively characterized by the theses of Aristotle and Boethius\footnote{For an overview of connexive logic, see \cite{mccall2012history,sep-logic-connexive}.}
	
\begin{description}
\item [{\textsc{Aristotle}:}] {$\neg(A\To\neg A)$ and $\neg(\neg A\To A)$ are valid.}
\item [{\textsc{Boethius}:}] $(A\To B)\To\neg(A\To\neg B)$ and $(A\To\neg B)\To\neg(A\To B)$ are valid.
\end{description}
	
However, some connexive logics allowed for satisfiable instances of $(A\To  \neg A)$, as well as simultaneously satisfiable instances of $(A\To B)$ and $(A\To\neg B)$. Kapsner took that to go against the spirit of the connexive enterprise. To be able to judge those cases out of bounds, he suggested to add two unsatisfiablity clauses:

\begin{description} \label{Unsats}
\item [{\textsc{UnSat1}:}] In no model,{ $(A\To  \neg A)$ is satisfiable, and neither is $(\neg A \To   A)$. }
\item [{\textsc{UnSat2}:}] {In no model $(A\To B)$ and $(A\To\neg B)$ are satisfiable simultaneously (for any $A$ and $B$).}
\end{description}
	
He called logics that satisfy Aristotle, Boethius and the UnSat clauses \textit{strongly connexive}, those that only satisfied Aristotle and Boethius \textit{weakly connexive}.
	
To call for such clauses that prescribe unsatisfiabilities, of course, is a quite uncommon move in non-classical logic, and Kapsner suggested, in a very tentative way, that there might be an idea worth exploring that tries to push the requirement into the object language. He wrote:
	
\begin{quote}
Now, one might wonder whether the criterion of strong connexivity can be expressed in some manner in the object language itself, given that \textsc{Aristotle} and \textsc{Boethius} are not up to this task. 
		
A look at an analogous problem might bring us to some interesting ideas here. Classical logic and most non-classical logics validate the principle of \emph{explosion}, $(A\wedge\neg A){\rightarrow} B$. Let us remind ourselves of a philosophical argument that is sometimes given in favor of the validity of $(A\wedge\neg A){\rightarrow} B$. This validity, it is claimed, is the most we can do to express, in the object language itself, the thought that a contradiction is unsatisfiable. 
		
In analogy to this use of explosion to express the unsatisfiability of any contradiction, we might try to ask that $(A\rightarrow\neg A)\rightarrow B$ should be valid, in order to express in the object language that $A\rightarrow\neg A$ is unsatisfiable (and similarly for the rest of the connexive theses). Call a logic that validates all of these schemata and satisfies all the requirements for strong connexivity \emph{superconnexive}. (\cite[p.143]{Kapsner2012strong})
\end{quote}
	
However, as he immediately went on to note, adding this to a system with substitutivity of logical equivalents quickly leads to trouble. Even though the option of giving up substitutivity has been explored before in the realm of connexive logics (see \cite{Priest1999}), Kapsner did not develop the idea of superconnexivity any further in the paper, and it has not found many friends since.

In this paper, we want to explore whether something of the idea can be salvaged, even if we insist on keeping substitutivity. We argue that this is possible, even though we have to modify the idea a little bit. 
	
In order for this to amount to an interesting route for further exploration for connexivists, we will, as a preparatory move, begin \S\ref{Problems} by stating the condition Kapnser proposed in \cite{Kapsner2012strong} a little more carefully and extending it to other connexive theses. We will then immediately take a closer look at the problems these principles engender. In \S\ref{SuperBot}, we will introduce our idea of how to avoid these problems, which involves the introduction of a bottom constant. We will then draw several conclusions about how these new principles interact with others. This will be followed by \S\ref{Explosion} in which we focus on the interplay between our new conception of superconnexivity and Explosion, the principle that inspired the idea of superconnexivity in the first place. \S\ref{Abelard}, then, takes some deeper looks at two lesser known principles (Abelard and Aristotle's second thesis) that sometimes get discussed in the connexive literature. In \S\ref{SuperBot AT BT}, we explore the exact connections between the new and the old principles, Aristotle and Boethius. Finally, \S\ref{Summary} sums up our findings, and discusses briefly some future directions.
	
\section{Superconnexivity and its Troubles, More Carefully Examined}\label{Problems}
	
We first examine Super-Aristotle, and then turn to Super-Boethius. Kapsner's suggestion was to capture the unsatisfiability of $A\To \neg A$ by the validity of $(A\To \neg A)\To B$, leaving the question of how to treat the unsatisfiability of $\neg A\To A$ open. Though it is rather obvious, let us spell it out and give the principles names to facilitate the upcoming discussion:
	
\begin{description}
\item [{\textsc{Super-Aristotle}:}]  {$(A\To\neg A)\To B$ and $(\neg A\To A)\To B$ are valid.}
\end{description}
	
It will turn out that there is a relevant difference in the behavior of these two variations, so they should be mentioned explicitly. Where we mean either the first or the second in particular, we will write SA1 and SA2, respectively.
	
We also want to suggest a way of giving a superconnexive version of Boethius / UnSat2, something Kapsner failed to do in \cite{Kapsner2012strong}. Here is our suggestion:
	
\begin{description}
\item [{\textsc{Super-Boethius}:}]  {$(A\To  B)\To ((A\To \neg B)\To C)$ and $(A\To \neg B)\To ((A\To B)\To C)$ are valid.\footnote{An anonymous referee asks why this, rather than what we will call \textsc{Super-Abelard} in \S\ref{Abelard}, is our suggested rendering of \textsc{UnSat2}. This is a very good question, and, really, the original sin here lies with Kapsner's \cite{Kapsner2012strong}. As he is one of the authors of the present papers, this does not get us off the hook completely. At the time of writing, and still today, Kapsner was unable to phrase an UnSat principle that was clearly the correlate of Boethius, but not of Abelard (he thought and still thinks that both are motivated by very similar intuitions that can be phrased as \textsc{UnSat2}). That we now focus on \textsc{Super-Boethius} rather than \textsc{Super-Abelard} reflects that connexive logic as a modern research area places more emphasis on Boethius than Abelard.}}
\end{description}
	
Again, a careful examination advises being explicit about the two different variations, we will call them SB1 and SB2 below.
	
\subsection{Problems: Super-Aristotle}
Here is the problem with Super-Aristotle already noted in \cite{Kapsner2012strong}: As an instance of the axiom, substituting $\neg (A\To \neg A)$ for $B$, we obtain $(A\To \neg A)\To \neg (A\To \neg A)$. This, however, is an instance of exactly the thing that is to be avoided in view of UnSat 1. 
	
While this is true we'd like to note that, in fact, the situation is worse still than that, if we have modus ponens in place (and we will assume throughout this essay that we do). Indeed, we have the following.
	
\begin{proposition}
$(A\To \neg A)\To B$, MP and SUB is trivial.
\end{proposition}
\begin{proof}
The proof runs as follows.
	
1 \ \ \ $(A\To \neg A)\To \neg (A\To \neg A)$ \hfill[SA1]
		
2 \ \ \ $((A\To \neg A)\To \neg (A\To \neg A))\To B$ \hfill[SA1]

3 \ \ \ $B$ \hfill[1, 2, MP]

\noindent This completes the proof.
\end{proof}

This assumes the first part of Super-Aristotle above, $(A\To\neg A)\To B$. Interestingly, and somewhat surprisingly, we do \emph{not} get the same sort of trouble with $(\neg A\To A)\To B$ if we assume only MP and SUB. 
	
\begin{proposition}
$(\neg A\To A)\To B$, MP and SUB is non-trivial.
\end{proposition}
\begin{proof}
Consider the following truth tables with $\mathbf{1}$ as the only designated value.
\begin{center}
{\small
\begin{tabular}{c|c}
& $\neg $    \\ \hline
$\mathbf{1}$ & $\mathbf{i}$  \\
$\mathbf{i}$ & $\mathbf{1}$\\
$\mathbf{0}$ & $\mathbf{1}$  \\  
\end{tabular}
\quad
\begin{tabular}{c|cccccc}
$\To$ & $\mathbf{1}$ & $\mathbf{i}$ & $\mathbf{0}$ \\   \hline
$\mathbf{1}$ & $\mathbf{1}$ & $\mathbf{0}$ & $\mathbf{0}$ \\ 
$\mathbf{i}$ & $\mathbf{0}$ & $\mathbf{1}$ & $\mathbf{0}$ \\  
$\mathbf{0}$ & $\mathbf{1}$ & $\mathbf{1}$ & $\mathbf{1}$ \\  
\end{tabular}
}
\end{center}
Then, $\neg A\To A$ always receives the value $\mathbf{0}$, and thus $(\neg A\To A)\To B$ is valid. Moreover, MP preserves the value $\mathbf{1}$. However, note that $\neg \neg p\To p$ receives the value $\mathbf{0}$ when $p$ is assigned the value $\mathbf{0}$, and thus the formula is not valid and the system is non-trivial.
\end{proof}
	
Still, we obtain the following results.
	
\begin{proposition}
$(\neg A\To A)\To B$, MP and SUB together with double negation elimination (DNE) and the rule of Transitivity (Trans.) is trivial.
\end{proposition}
	
\begin{proof}
The proof runs as follows.

1 \ \ \ $\neg \neg (\neg A\To A)\To (\neg A\To A)$ \hfill[DNE]
		
2 \ \ \ $(\neg A\To A)\To \neg (\neg A\To A)$ \hfill[SA2]
		
3 \ \ \ $\neg \neg (\neg A\To A)\To \neg (\neg A\To A)$ \hfill[1, 2, Trans.]

4 \ \ \ $(\neg \neg (\neg A\To A)\To \neg (\neg A\To A))\To B$ \hfill[SA2]
		
5 \ \ \ $B$ \hfill[3, 4, MP]

\noindent This completes the proof.
\end{proof}
	
\subsection{Problems: Super-Boethius}
For Super-Boethius, we will be in the same kind of trouble for one of the two versions. 

\begin{proposition}
$(A\To  B)\To ((A\To \neg B)\To C)$, MP and SUB is trivial.
\end{proposition}
\begin{proof}
The proof runs as follows.
		
1 \ \ \ $(A\To  B)\To ((A\To \neg B)\To C)$ \hfill[SB1]
		
2 \ \ \ $((A\To  B)\To ((A\To \neg B)\To C))\To (((A\To  B)\To \neg ((A\To \neg B)\To C))\To D)$ \hfill[SB1]
		
3 \ \ \ $((A\To  B)\To \neg ((A\To \neg B)\To C))\To D$ \hfill[1, 2, MP]
		
4 \ \ \ $((A\To  B)\To ((A\To \neg B)\To C))\To (((A\To  B)\To \neg ((A\To \neg B)\To C))\To \neg D)$ \hfill[SB1]
		
5 \ \ \ $((A\To  B)\To \neg ((A\To \neg B)\To C))\To \neg D$ \hfill[1, 4, MP]
		
6 \ \ \ $E$ \hfill[3, 5, SB1, MP]
		
\noindent This completes the proof.
\end{proof}
	
We do not know yet if the same triviality proof is possible with the other version, but we will immediately obtain the following by having double negation introduction (DNI).
	
\begin{proposition}
$(A\To  \neg B)\To ((A\To B)\To C)$, MP and SUB with DNI is trivial.
\end{proposition}
\begin{proof}
It suffices to derive $(A\To\neg A)\To B$, and this can be done as follows.
		
1 \ \ \ $(A\To  \neg \neg A)\To ((A\To \neg A)\To B)$ \hfill[SB2]
		
2 \ \ \ $(A\To\neg A)\To B$ \hfill[1, DNI, MP]
		
\noindent This completes the proof.
\end{proof}
	
In sum, the superconnexive versions of Aristotle and Boethius theses either immediately trivialize with MP and SUB, or after adding at most a couple of plausible (though not uncontested) principles like the double negation laws and transitivity.

\section{An alternative: Super-Bot-Connexivity}\label{SuperBot}
	
So, we can not use superconnexivity as it stands, unless we are willing to do without the principles involved in the proofs in the last section. As announced above, we want to explore a modification of the idea that is logically not as destructive, but does preserve the conceptual core of the original idea.
	
\subsection{Our suggestion}
Our idea is to replace the arbitrary $B$ with a bottom, as follows:

\begin{description}
\item [{\textsc{Super-Bot-Aristotle}:}]  {$(A\To\neg A)\To \bot$ and $(\neg A\To A)\To \bot$ are valid.}
\end{description}
We will refer to the first as S$\bot$A1, to the second as S$\bot$A2.
	
\medskip
	
The hope is that by taming the behavior of $\bot$ just enough to retain a sense of absurdity but avoid triviality of the logical systems the principle is added to, we can make sense of the original super-connexive idea in this way.
	
Here is how the corresponding forms of Super-Boethius would look like in accordance with our suggestion above:
\begin{description}
\item [{\textsc{Super-Bot-Boethius}:}]  {$(A\To B)\To ((A\To \neg B)\To \bot)$ and $(A\To \neg  B)\To ((A\To B)\To \bot)$ are valid.}
\end{description}
We will refer to the first as S$\bot$B1, to the second as S$\bot$B2.
	
\subsection{What $\bot$ is (not)}
	
What properties can we give $\bot$, which ones can't we have? Let us begin with another triviality result involving $\bot\To A$, which we will call EFQ (\textit{ex falso quodlibet}).

\begin{proposition}
$\bot\To A$, S$\bot$A1, MP and SUB is trivial.
\end{proposition}
\begin{proof}
The proof runs as follows.
		
1 \ \ \ $\bot\To \neg \bot$ \hfill[EFQ]
		
2 \ \ \ $(\bot\To \neg \bot)\To \bot$ \hfill[S$\bot$A1]

3 \ \ \ $\bot$ \hfill[1, 2, MP]
		
4 \ \ \ $A$ \hfill[3, EFQ, MP]

\noindent This completes the proof.
\end{proof}
	
Therefore, we \emph{cannot} have EFQ. Given the close proximity of connexive and relevant logic, at least in terms of the researchers who are interested in these topics, it is interesting to see that EFQ will be ruled out in view of the above triviality result. 
	
But does that not go against the original motivation of chastising the people who treat $A\To\neg A$ as satisfiable? Not necessarily, as long as we make sure that $\bot$ is something truly absurd. As Graham Priest reminds us (\cite[p.422]{priest1998so}), if you are forced to concede that, for example, you are a fried egg, then that should be a bad enough deterrent. The exact choice of absurdity will depend on taste and context, but a basic requirement must be that $\bot$ is never satisfiable. 

One possible requirement is to have that $\bot \models  A$. Given standard assumptions about logical consequence, it is in fact hard to avoid this if $\bot$ is unsatisfiable. Luckily, this requirement will not lead to triviality, as we will observe below. But there is an immediate consequence of this which we should draw attention to right away:
	
\begin{proposition}
$\bot \models A$, S$\bot$A1, MP, SUB and the deduction theorem is trivial.
\end{proposition}
\begin{proof}
If we have $\bot \models A$ and the deduction theorem, then we obtain $\bot\To A$, but this will trivialize with S$\bot$A1, MP, SUB, as observed above.
\end{proof}
	
Thus, the deduction theorem needs to be given up, if we require $\bot \models A$ to express the thought that $\bot$ is absurd.
	
\subsection{Proof of Concept: CC1}
	
To see that the Super-Bot-principles can be satisfied, let us go back to the logic {\bf CC1}, characterized by Angell's four-valued logic (cf. \cite{AngellPropositional,McCall1966}), to which we add a bottom, which can be defined as $\neg (p\To p)$ for some $p$, that always takes value 4:
\begin{center}
{\small
\begin{tabular}{c|cccc|c} 
$\land$ & 1 & 2 & 3 & 4 & $\neg$\\ \hline
1     & 1 & 2 & 3 & 4 & 4\\
2     & 2 & 1 & 4 & 3 & 3\\
3     & 3 & 4 & 3 & 4 & 2\\
4     & 4 & 3 & 4 & 3 & 1
\end{tabular}	
\begin{tabular}{c|cccc|c} 
$\To$ & 1 & 2 & 3 & 4 & $\bot$\\ \hline
1     & 1 & 4 & 3 & 4 & 4\\
2     & 4 & 1 & 4 & 3 & 4\\
3     & 1 & 4 & 1 & 4 & 4\\
4     & 4 & 1 & 4 & 1 & 4
\end{tabular}
}
\end{center}
The designated values are 1 and 2.\footnote{Why do we choose a bottom that always takes value 4, and not value 3? Simply because it works better in giving us desirable validities.}
	
Kapsner has noted in \cite{Kapsner2012strong} that {\bf CC1} is a case of a strongly connexive logic, but that it is not super-connexive. It turns out, however, that the Super-Bot versions of Aristotle and Boethius are valid.

\begin{remark}
Formally, the above point can be made easily by considering classical logic, where the arrow is understood as the bi-conditional connective. But of course, one of the golden rules of connexive logic is that one should avoid turning a conditional into a bi-conditional. In {\bf CC1}, we can observe that $(p\To q)\To (q\To p)$ is \emph{not} valid. Indeed, assign 3 and 1 to $p$ and $q$, respectively.
\end{remark}

\subsection{Implications of Super-Bot-Connexivity}
	
We now turn to observe some consequences of Super-Bot-Connexivity.
	
First, we need to ascertain that weakly connexive logics are ruled out, or else Super-Bot-Connexivity can make no claim to be a restatement of the demands of strong connexivity. Fortunately, if we assure that $\bot$ can not receive a designated value, then we get this result rather easily.
	
\begin{proposition}
Assume super-bot-connexivity, MP and SUB. Then, weakly connexive logics prove $\bot$.
\end{proposition}
	
\begin{proof}
Given a weakly connexive logic, there is an instance of UnSat1 or UnSat 2 being violated. Take that instance by SUB, and apply the MP with the suitable super-bot-connexivity. Then, the desired result follows.
\end{proof}
	
But more than weak connexivity is ruled out. Some of the next items we prove put super-bot-connexivity in close proximity to relevant concerns, others to those of constructivists.
	
Two examples for principles that are suspect to the relevant logicians and that are prevented by super-bot-connexivity are Weakening and (a form of) Explosion:
	
\begin{proposition}\label{Prop Weakening}
Weakening, S$\bot$A2, MP and SUB proves $\bot$.
\end{proposition}
\begin{proof}
The proof runs as follows.
		
1 \ \ \ $((\neg A\To A)\to \bot)\To (\neg ((\neg A\To A)\to \bot)\To ((\neg A\To A)\to \bot))$ \hfill[Weakening]
		
2 \ \ \ $(\neg A\To A)\to \bot$ \hfill[S$\bot$A2]
		
3 \ \ \ $\neg ((\neg A\To A)\to \bot)\To ((\neg A\To A)\to \bot)$ \hfill[1, 2, MP]
		
4 \ \ \ $(\neg ((\neg A\To A)\to \bot)\To ((\neg A\To A)\to \bot))\To \bot$ \hfill[S$\bot$A2]
		
5 \ \ \ $\bot$ \hfill[3, 4, MP]
		
\noindent This completes the proof.
\end{proof}
	
\begin{proposition}\label{prop:SbotB12ECQ}
$A\To (\neg A\To B)$, S$\bot$B1 or S$\bot$B2, MP and SUB proves $\bot$.
\end{proposition}
\begin{proof}
The proof runs as follows.
		
1 \ \ \ $A\To (\neg A\To A)$ \hfill[ECQ]
		
2 \ \ \ $(A\To (\neg A\To A))\To (\neg (A\To (\neg A\To A))\To B)$ \hfill[ECQ]
		
3 \ \ \ $\neg (A\To (\neg A\To A))\To B$ \hfill[1, 2, MP]
		
4 \ \ \ $(A\To (\neg A\To A))\To (\neg (A\To (\neg A\To A))\To \neg B)$ \hfill[ECQ]
		
5 \ \ \ $\neg (A\To (\neg A\To A))\To \neg B$ \hfill[1, 4, MP]
		
6 \ \ \ $(\neg (A\To (\neg A\To A))\To B)\To ((\neg (A\To (\neg A\To A))\To \neg B)\To \bot)$ \hfill[S$\bot$B1]
		
7 \ \ \ $\bot$ \hfill[3, 5, 6, MP]
		
\noindent The case with S$\bot$B2 is similar.
\end{proof}
	
\begin{remark}\label{rmk:SbotB12ECQ}
We can also replace $A\To (\neg A\To B)$ by $(A\land \neg A)\To B$ and $((A\land B)\To C)\To ((A\To (B\To C)))$ to obtain the derivablity of $\bot$. We will discuss other forms of expressing Explosion in the next section.
\end{remark}
	
Let us add a few more implications of super-bot-connexivity.
	
\begin{proposition}
$A\To \top$, S$\bot$A2, MP and SUB proves $\bot$.
\end{proposition}
\begin{proof}
The proof runs as follows.

1 \ \ \ $\neg \top\To \top$ \hfill[$A\To \top$]
		
2 \ \ \ $(\neg \top\To \top)\To \bot$ \hfill[S$\bot$A2]

3 \ \ \ $\bot$ \hfill[1, 2, MP]

\noindent This completes the proof.
\end{proof}
	
\begin{remark}
Exactly the same proof also establishes that $\neg \top\To A$, S$\bot$A2, MP and SUB proves $\bot$.
\end{remark}
	
\begin{remark}
The above result can be seen as showing a tension between Super-Bot-Aristotle and \textit{ubiquitous truths}, a notion introduced by Andr\'e Fuhrmann in \cite{Fuhrmann1990models} in the context of developing modal logics that expand relevant logic.\footnote{See also \cite{Standefer2020actual} for a recent discussion on ubiquitous truths.} Note, also the difference between $A\To \top$ and $\top$. The former is in tension with Super-Bot-Aristotle, but the latter is not. This might be seen as another conceptual link to relevant logic.
\end{remark}
	
On the other hand, here is something that constructive logicians have been known to reject and that is similarly ruled out by super-bot-connexivity, 
	
\begin{proposition}
Peirce's law, S$\bot$A1, MP and SUB proves $\bot$.
\end{proposition}
\begin{proof}
The proof runs as follows.
		
1 \ \ \ $(\bot\To \neg \bot)\To \bot$ \hfill[S$\bot$A1]
		
2 \ \ \ $((\bot\To \neg \bot)\To \bot)\To \bot$ \hfill[Peirce's law]

3 \ \ \ $\bot$ \hfill[1, 2, MP]

\noindent This completes the proof.
\end{proof}
	
Lastly, here are two things that have few enemies, but are incompatible with super-bot-connexivity. They are the deduction theorem and conjunction elimination, principles that, their inoffensiveness not\-with\-stand\-ing, often fail in connexive settings.
	
\begin{proposition}
The deduction theorem, S$\bot$A2, MP and SUB proves $\bot$. 
\end{proposition}
\begin{proof}
An immediate corollary to Prop.~\ref{Prop Weakening}, since weakening is derivable by the deduction theorem.
\end{proof}
	
\begin{proposition}
Conjunction elimination, S$\bot$B1 or S$\bot$B2, MP and SUB proves $\bot$.
\end{proposition}
\begin{proof}
By considering instances of conjunction elimination, namely $(A{\land} \neg A)\To A$ and $(A{\land} \neg A)\To \neg A$.
\end{proof}
	
\begin{remark}
The prime examples of systems of connexive logic with both the deduction theorem and conjunction elimination are the weakly connexive systems {\bf C}, devised in \cite{Wansing2005} by Heinrich Wansing and {\bf CN}, due to John Cantwell (\cite{Cantwell2008}). See \cite[\S4.3]{OmoriWansingAiML20} for a discussion on UnSat 1 and 2 in these systems.
\end{remark}
	
\section{Super-Bot-Connexivity and Explosion}\label{Explosion}

Let us make some remarks on the relationship between super-bot-connexivity and Explosion, a relationship that seems pertinent to focus on given the origin of the idea of super-bot-connexivity. After all, the original idea of super-connexivity was explicitly modeled after the idea of Explosion.
	
We have different ways of capturing Explosion to consider. Kapsner started out from the arrow form involving conjunction:
	
\begin{description}
\item [{\textsc{ECQ$_{\land}$}:}]  $(A\land \neg A)\To B$.
\end{description}
We can also consider the following form which does not involve conjunctions:
	
\begin{description}
\item [{\textsc{ECQ$_{\to}$}:}]  $A\To (\neg A\To B)$.
\end{description}
	
We have seen above in Proposition~\ref{prop:SbotB12ECQ} that $A\To (\neg A\To B)$ is incompatible with Super-Bot-Boethius. Given Exportation, the same goes for $(A\land \neg A)\To B$ (recall Remark~\ref{rmk:SbotB12ECQ}), but there is also another proof that makes use of S$\bot$A1.

\begin{proposition}
$(A\land \neg A)\To B$, S$\bot$A1, MP and SUB proves $\bot$.
\end{proposition}
\begin{proof}
The proof runs as follows.

1 \ \ \ $(A\land \neg A)\To \neg (A\land \neg A)$ \hfill[ECQ$_{\land}$]
		
2 \ \ \ $((A\land \neg A)\To \neg (A\land \neg A))\To \bot$ \hfill[S$\bot$A1]

3 \ \ \ $\bot$ \hfill[1, 2, MP]

\noindent This completes the proof.
\end{proof}
	
Therefore, there are some tensions between super-bot-connexivity and the above forms of Explosion. However, given the availability of $\bot$ in our proposal, we can also consider these versions of Explosion, which we name ECF (\textit{ex contradictione falsum}):
	
\begin{description}
\item [{\textsc{ECF$_{\land}$}:}]  $(A\land \neg A)\To \bot$.
\item [{\textsc{ECF$_{\to}$}:}]  $A\To (\neg A\To \bot)$.
\end{description}
	
Now, in {\bf CC1}, we do have $(A\land \neg A)\To \bot $, so \textsc{ECF$_{\land}$} is compatible with the super-bot principles. On the other hand, in {\bf CC1}, we don't have $A\To (\neg A\To \bot)$, but we do see it satisfied in the logic that has the following truth tables and $1$ as the designated value.\footnote{This is a logic closely related to {\bf CC1}, and studied in more depth in \cite{AM3}, but here it only serves as a way to show the consistency of these ideas.}
\begin{center}
{\small
\begin{tabular}{ccc}
\begin{tabular}{c|c}
& $\neg $    \\ \hline
$\mathbf{1}$ & $\mathbf{0}$  \\
$\mathbf{i}$ & $\mathbf{0}$\\
$\mathbf{0}$ & $\mathbf{1}$  \\  
\end{tabular}
&
\begin{tabular}{c|cccccc}
$\land$ & $\mathbf{1}$ & $\mathbf{i}$ & $\mathbf{0}$ \\   \hline
$\mathbf{1}$ & $\mathbf{1}$ & $\mathbf{i}$ & $\mathbf{0}$ \\ 
$\mathbf{i}$ & $\mathbf{i}$ & $\mathbf{i}$ & $\mathbf{0}$ \\  
$\mathbf{0}$ & $\mathbf{0}$ & $\mathbf{0}$ & $\mathbf{i}$ \\  
\end{tabular}
&
\begin{tabular}{c|cccccc}
$\To$ & $\mathbf{1}$ & $\mathbf{i}$ & $\mathbf{0}$ \\   \hline
$\mathbf{1}$ & $\mathbf{1}$ & $\mathbf{i}$ & $\mathbf{0}$ \\ 
$\mathbf{i}$ & $\mathbf{1}$ & $\mathbf{1}$ & $\mathbf{0}$ \\  
$\mathbf{0}$ & $\mathbf{0}$ & $\mathbf{0}$ & $\mathbf{1}$ \\  
\end{tabular}
\end{tabular}
}
\end{center}
Therefore, super-bot-connexivity and the bot-versions of explosion are compatible.
	
Philosophically speaking, how to make sense of all of this requires more space than we have here, but here are some first thoughts: There is a close connection between connexive and relevant ideas, at the very least sociologically speaking. To see that super-bot-connexivity and very few other assumptions rule out $A\To (\neg A\To B)$ is, in that sense, maybe a welcome outcome for many in the respective communities (especially as Weakening is also ruled out). Whether the Bot-versions of Explosion are acceptable to relevantists will, probably, be dependant on a more fully fleshed out philosophical explanation of what $\bot$ means, but we will leave that question for future study.
	
On the other hand, that $(A\land \neg A)\To \bot $ and $A\To (\neg A\To \bot)$ are compatible with the proposal seems to be in harmony with the original story of the super-connexive idea, which was to model the object language representation of the UnSat-principle after the use Explosion is put to by those who want to rule out the sastisfiability of contradictions. It would, maybe, seem strange if the two ideas could not be combined in some way, and we find that the $\bot$-versions of these principles are the way to achieve this.\footnote{The philosophical position that holds that contradictions are satisfiable but asks for strong connexivity is, of course, possible, but it seems doubtful whether it will find many adherents.}
	
\section{An interlude: Super-Bot versions of Abelard and Aristotle 2nd} \label{Abelard}
	
There are two more principles that often (though not in \cite{Kapsner2012strong}) come up in discussions of connexive logic, namely Abelard's thesis  and Aristotle's second thesis:
	
\begin{description}
\item [{\textsc{Abelard}:}]  {$\neg ((A\To B)\land (A\To \neg B))$ is valid;}
\end{description}

\begin{description}
\item [{\textsc{Aristotle's Second Thesis}:}]  {$\neg ((A\To B)\land (\neg A\To B))$ is valid.}
\end{description}
Some of the recent discussions that are related to these principles include \cite{Rott2019difference,Crupi2020evidential}. We do not wish to enter the debate whether these are plausible or whether they should or should not be seen as part of the connexive canon. However, for the sake of completeness of our discussion, we would like to investigate whether there are Super-Bot treatments of these principles, and how they might interact with the others.

As we just mentioned, these are not discussed in \cite{Kapsner2012strong}. However, it seems natural to say that UnSat2 captures the intuition that makes Abelard plausible, while Aristotle's Second Thesis would call for a new unsatisfiability requirement. We would like to suggest the following:

\begin{description} \label{Unsat3}
\item [{\textsc{UnSat3}:}] {In no model $(A\To B)$ and $(\neg A\To B)$ are satisfiable simultaneously (for any $A$ and $B$).}
\end{description}

Here is what one could expect a superconnexive treatment of these theses to look like:
	 
\begin{description}
\item [{\textsc{Super-Abelard}:}]  {$((A\To B)\land (A\To \neg B))\To C$ is  valid;}
\item [{\textsc{Super-Aristotle2}:}]  {$((A\To B)\land (\neg A\To B))\To C$ is  valid.}
\end{description}

Although it is not clear to us at the moment if the above principles trivialize only with MP and SUB, we do obtain triviality if we have the rule of Adjunction for the case with Super-Abelard, as follows.

\begin{proposition}
The rule of Adjunction (Adj.), Super-Abelard, MP and SUB is trivial.
\end{proposition}
\begin{proof}
The proof runs as follows.

1 \ \ \ $((A\To B)\land (A\To \neg B))\To C$ \hfill[Super-Abelard]
		
2 \ \ \ $((A\To B)\land (A\To \neg B))\To \neg C$ \hfill[Super-Abelard]
		
3 \ \ \ $(((A\To B)\land (A\To \neg B))\To C)\land (((A\To B)\land (A\To \neg B))\To \neg C)$ \hfill[1, 2, Adj.]
		
4 \ \ \ $((((A\To B)\land (A\To \neg B))\To C)\land (((A\To B)\land (A\To \neg B))\To \neg C))\To D$ \hfill[Super-Abelard]

5 \ \ \ $D$ \hfill[3, 4, MP]
		
\noindent This completes the proof.
\end{proof}

So, maybe we are not facing quite the same need as we did in the cases of Super-Aristotle and Super-Boethius to consider variations of Super-Abelard and Super-Aristotle2. Nonetheless here is how a super-bot-connexive approach to these principles might look: 

\begin{description}
\item [{\textsc{Super-Bot-Abelard}:}]  {$((A\To B)\land (A\To \neg B))\To \bot$ is  valid;}
\item [{\textsc{Super-Bot-Aristotle2}:}]  {$((A\To B)\land (\neg A\To B))\To \bot$ is valid.\footnote{A referee interestingly wonders if ``conditionalized'' versions of \textsc{Super-Aristotle2} and \textsc{Super-Bot-Aristotle2}, namely the following principles, make sense. 
\begin{itemize}
\item $(A\To B)\To ((\neg A\To B)\To C)$ is valid.
\item $(A\To B)\To ((\neg A\To B)\To \bot)$ is valid.
\end{itemize}
This seems to be very interesting, possibly having some connections to the system introduced in \cite{Omori2016Francez}, motivated by \cite{Francez2016natural}. Further investigations, however, will be left for another occasion.}}
\end{description}
	
Let us now briefly explore the simple implications of the bot versions of Abelard's thesis and Aristotle 2nd thesis combined with some familiar axioms/rules involving conjunction and/or disjunction.
	
\subsection{Super-Bot-Abelard}
We start with a triviality result.
	
\begin{proposition}
$\bot\To A$, the rule of Adjunction, Super-Bot-Abelard, MP and SUB is trivial.
\end{proposition}
\begin{proof}
The proof runs as follows.

1 \ \ \ $\bot\To A$ \hfill[EFQ]
		
2 \ \ \ $\bot\To \neg A$ \hfill[EFQ]
		
3 \ \ \ $(\bot\To A)\land (\bot\To \neg A)$ \hfill[1, 2, Adj.]

4 \ \ \ $\bot$ \hfill[3, Super-Bot-Abelard, MP]
		
5 \ \ \ $A$ \hfill[4, EFQ, MP]

\noindent This completes the proof.
\end{proof}
	
\begin{remark}
By replacing $\bot$ by $\neg \top$, and removing the last line, we can see that $\neg\top\To A$, the rule of Adjunction, Super-Bot-Abelard, MP and SUB proves $\bot$.
\end{remark}
	
We also have a tension with conjunction elimination.
	
\begin{proposition}
Conjunction elimination, the rule of Adjunction, Super-Bot-Abelard, MP and SUB proves $\bot$.
\end{proposition}
\begin{proof}
The proof runs as follows.

1 \ \ \ $(A{\land} \neg A)\To A$ \hfill[Conj. elim.]
		
2 \ \ \ $(A{\land} \neg A)\To \neg A$ \hfill[Conj. elim.]

3 \ \ \ $((A{\land} \neg A)\To A)\land ((A{\land} \neg A)\To \neg A)$ \hfill[1, 2, Adj.]
		
4 \ \ \ $\bot$ \hfill[3, Super-Bot-Abelard, MP]

\noindent This completes the proof.
\end{proof}
	
\subsection{Super-Bot-Aristotle 2nd}
	
\begin{proposition}
$A\To \top$, the rule of Adjunction, Super-Bot-Aristotle 2nd, MP and SUB proves $\bot$.
\end{proposition}
\begin{proof}
The proof runs as follows.
		
1 \ \ \ $A\To \top$ \hfill[$A\To \top$]
		
2 \ \ \ $\neg A\To \top$ \hfill[$A\To \top$]

3 \ \ \ $(A\To \top)\land (\neg A\To \top)$ \hfill[1, 2, Adj.]
		
4 \ \ \ $\bot$ \hfill[3, Super-Bot-Aristotle 2nd, MP]

\noindent This completes the proof.
\end{proof}
	
\begin{remark}
By replacing $\top$ by $\neg\bot$, we can see that $A\To \neg\bot$, the rule of Adjunction, Super-Bot-Aristotle 2nd, MP and SUB proves $\bot$.
\end{remark}
	
Finally, if we consider a language with disjunction, then we obtain the following result.
\begin{proposition}
Disjunction introduction, the rule of Adjunction, Super-Bot-Aristotle 2nd, MP and SUB proves $\bot$.
\end{proposition}
\begin{proof}
The proof runs as follows.
	
1 \ \ \ $A\To (A{\lor} \neg A)$ \hfill[Disj. intro.]
	
2 \ \ \ $\neg A\To (A{\lor} \neg A)$ \hfill[Disj. intro.]

3 \ \ \ $(A\To (A{\lor} \neg A))\land (\neg A\To (A{\lor} \neg A))$ \hfill[1, 2, Adj.]
	
4 \ \ \ $\bot$ \hfill[3, Super-Bot-Aristotle 2nd, MP]

\noindent This completes the proof.
\end{proof}
	
\section{Super-Bot-Connexivity, Aristotle and Boethius} \label{SuperBot AT BT}

Let us, to bring things to a preliminary end, consider how the super-bot-connexive theses relate to Aristotle and Boethius. The claim was that the former capture the intuitions that motivate the latter, in that sense one might assume that there are strong relations between them that rely on well-established principles only. On the other hand, if they turn out to be only loosely connected, there might be a chance to find systems that only satisfy the super-bot-connexive theses and so leaves some air to get a maximal amount of other desirable properties.\footnote{This might be seen as a form of \textit{Kapsner-strong connexivity} in the terminology of \cite{Estrada-G0nzalez2016}.}
	
How close is this connection, then? There is an interesting way in which we can make it extremely close, indeed. Take another look at Aristotle and Super-Bot-Aristotle: 
\begin{description}
\item [{\textsc{Aristotle}:}]  $\neg (A\To\neg A)$
\item [{\textsc{Super-Bot-Aristotle}:}]  $(A\To\neg A)\To \bot$
\end{description}
Juxtaposing the principles like that might remind you of the intuitionists' handling of negation. Indeed, Super-Bot-Aristotle just looks like Aristotle, where the outer negation is understood along intuitionistic lines. Let us push this line a bit further: What if we adopt that view of negation wholesale and consider $\neg A$ as defined as $A\To\bot$?
	
It is interesting to note that connexivity has been investigated in relation to many kinds of negation, but no discussion we are aware of relates the connexive principles to the intuitionistic understanding of negation.\footnote{A referee points out that this is not only the understanding of intuitionistic, but also of minimal logicians. We are happy to acknowledge that, but, at the same time, decided to keep using our term here, simply because intuitionistic negation is both historically prior and more well known than minimal negation.} This is a very interesting lacuna that we intend to explore in more detail, starting in \cite{AM3}. Here we just want to remark that, given the intuitionistic understanding of negation, Aristotle and Super-Bot-Aristotle are the very same thing, and the same goes for Boethius and Super-Bot-Boethius.
	
But consider the case in which we don't think of negation this way. What can we say then? Here are some observations.
	
\begin{proposition}
S$\bot$A1, MP and SUB, together with $\neg\bot$ and the rule of Contraposition proves AT1.
\end{proposition}
\begin{proof}
The proof runs as follows.
		
1 \ \ \ $(A\To \neg A)\To \bot$ \hfill[S$\bot$A1]
		
2 \ \ \ $\neg \bot\To \neg (A\To \neg A)$ \hfill[1, Contra.]

3 \ \ \ $\neg \bot$ \hfill[$\neg\bot$]
		
4 \ \ \ $\neg (A\To \neg A)$ \hfill[2, 3, MP]

\noindent This completes the proof.
\end{proof}
	
\begin{remark}
Note that our assumption is not trivializing in view of the truth table for {\bf CC1}.
\end{remark}
	
By a very similar proof, we also obtain the following.
	
\begin{proposition}
S$\bot$A2, MP and SUB, together with $\neg\bot$ and the rule of Contraposition proves AT2.
\end{proposition}
	
Moreover, the same trick works for Abelard's thesis as well as Aristotle's 2nd thesis.
	
\begin{proposition}
Super-Bot-Abelard, MP and SUB, together with $\neg\bot$ and the rule of Contraposition proves Abelard's thesis.
\end{proposition}
\begin{proof}
The proof runs as follows.
	
1 \ \ \ $((A\To B)\land(A\To \neg B))\To \bot$ \hfill[Super-Bot-Abelard]
	
2 \ \ \ $\neg \bot\To \neg ((A\To B)\land(A\To \neg B))$ \hfill[1, Contra.]

3 \ \ \ $\neg \bot$ \hfill[$\neg\bot$]
	
4 \ \ \ $\neg ((A\To B)\land(A\To \neg B))$ \hfill[2, 3, MP]

\noindent This completes the proof.
\end{proof}
	
Again, by a very similar proof, we also obtain the following.
\begin{proposition}
Super-Bot-Aristotle 2nd, MP and SUB, together with $\neg\bot$ and the rule of Contraposition proves Aristotle's 2nd thesis.
\end{proposition}
	
It is less obvious for Boethius theses, but we do have the following results.
	
\begin{proposition}
S$\bot$B1, MP and SUB, together with $\neg\bot$ and the rule of Contraposition proves the rule form of BT1, namely $A\To B\vdash \neg (A\To \neg B)$.
\end{proposition}
\begin{proof}
The proof runs as follows.
	
1 \ \ \ $A\To B$ \hfill[sup.]
	
2 \ \ \ $(A\To \neg B)\To \bot$ \hfill[1, S$\bot$B1, MP]
	
3 \ \ \ $\neg \bot\To \neg (A\To \neg B)$ \hfill[2, Contra.]

4 \ \ \ $\neg \bot$ \hfill[$\neg\bot$]
	
5 \ \ \ $\neg (A\To \neg B)$ \hfill[3, 4, MP]

\noindent This completes the proof.
\end{proof}
	
By a very similar proof, we also obtain the following.
\begin{proposition}
S$\bot$B2, MP and SUB, together with $\neg\bot$ and the rule of Contraposition proves the rule form of BT2, namely $A\To \neg B\vdash \neg (A\To B)$.
\end{proposition}
	
\begin{remark}
The above results only establish the rule form of BT1 and BT2, not the arrow form we are focusing on in this piece. However, it is worth noting that these are sometimes also referred to as Boethius theses in the literature (see e.g. \cite[p.416]{mccall2012history}). Moreover, this is basically the form of Boethius theses discussed by Claudio Pizzi since \cite{Pizzi1977}, and also by Priest in \cite{Priest1999}. In fact, this is even the case in some of the systems following the recipe suggested by Wansing in \cite{Wansing2005} that validates BT1 and BT2 by tweaking the falsity condition for various kinds of conditional, even with very weak ones (cf. \cite{Omori2016Unilog5,wansing2019connexive}). See, for example, the main system discussed in \cite{kapsner2017counterfactuals} as well as the system {\bf cCL} in \cite{wansing2019connexive}.
\end{remark}

\section{Summary}\label{Summary}
	
In this paper, we proposed to take a fresh look at the notion of superconnexivity. In the form it was originally proposed, it has severe problems, and we spent some time pinpointing them more carefully than had previously been done (\S\ref{Problems}).
	
However, we were able to make a suggestion that holds much more promise, namely to replace the arbitrary formula in the superconnexive theses by a bottom constant (\S\ref{SuperBot}). Properly constrained (but not too much), this move to what we call super-bot-connexivity has a credible claim to express the idea of strong connexivity. We pointed out some of the additional costs this move incurs, some of which might not be too unpleasant for certain philosophical outlooks (e.g. the loss of Weakening for relevantists, or the loss of Peirce's law for constructivists).
	
We spent \S\ref{Explosion} on taking a deeper look at the connections between the new proposal and the principle of Explosion, and went on in \S\ref{Abelard} to broaden our proposal to also include Super-Bot verions of Abelard and Aristotle's second thesis. Lastly, in \S\ref{SuperBot AT BT}, we explored how the new principles interact with the original connexive principles, Aristotle and Boethius.
	
In sum, it seems to us that super-bot-connexivity is a useful and interesting idea that has some claim to capture the ideas of the unsatisfiability requirements in the object language itself, which was the original aim of superconnexivity. Whether it really does will require more philosophical analysis than we were able to provide here, and we leave that for future research.

There are a number of other directions to explore, as well. We will here only mention two of them. First, given that we are now able to capture the unsatisfiability requirements in the object language, we can revisit various ways to capture Humble connexivity, introduced in \cite{Kapsner2019humble}. Second, we may consider the intuitionist's way of dealing with negation in the context of connexive logic, as we briefly touched on in \S\ref{SuperBot AT BT}. 
	
It remains to be seen how rich the idea of super-bot-connexivity is, and we hope some readers will be motivated to join the authors to continue with the development of connexive logics along these lines.



	\bibliographystyle{eptcs}
	\bibliography{biblio}
\end{document}